\documentclass[11pt,reqno]{amsart}
\allowdisplaybreaks[4]
\usepackage{graphicx} 
\usepackage{epsfig}
\usepackage{amssymb}
\usepackage{amsmath}
\usepackage{cite}
\newtheorem{theorem}{Theorem}

\newtheorem{proposition}[theorem]{Proposition}

\newtheorem{lemma}[theorem]{Lemma}
\newcommand{\be}{\begin{equation}}
\newcommand{\ee}{\end{equation}}
\newcommand{\bea}{\begin{eqnarray}}
\newcommand{\eea}{\end{eqnarray}}
\newcommand{\ba}{\begin{array}}
\newcommand{\ea}{\end{array}}
\newcommand{\bean}{\begin{eqnarray*}}
\newcommand{\eean}{\end{eqnarray*}}

\newcommand{\pa}{\partial}

\begin{document}

\title{The ``ghost" symmetry in the CKP hierarchy}
\author{Jipeng Cheng\dag, Jingsong He$^*$ \ddag }
\dedicatory {  \dag \ Department of Mathematics, China University of
Mining and Technology, Xuzhou, Jiangsu 221116,
P.\ R.\ China\\
\ddag \ Department of Mathematics, Ningbo University, Ningbo,
Zhejiang 315211, P.\ R.\ China }

\thanks{$^*$Corresponding author. Email: hejingsong@nbu.edu.cn.}
\begin{abstract}
In this paper, we systematically study the ``ghost" symmetry in the CKP hierarchy
through its actions on the Lax operator, dressing operator, eigenfunctions and the tau function.
In this process, the spectral representation of the eigenfunction is developed and the
squared eigenfunction potential is investigated.  \\
\textbf{Keywords}:  CKP
hierarchy, squared eigenfunction symmetry, spectral representation, squared eigenfunction potential.\\
\textbf{PACS}: 02.30.Ik\\
\textbf{2010 MSC}: 35Q53, 37K10, 37K40
\end{abstract}
\maketitle
\section{Introduction}
In this paper, given any pseudo-differential operator $A=\sum_ia_i\pa^i$ with $\pa=\pa_x$ and any function $f$,
\begin{eqnarray}
A_+=\sum_{i\geq 0}a_i\pa^i,\ A_-=\sum_{i< 0}a_i\pa^i,\ {\rm Res}(\sum_ia_i\pa^i)=a_{-1},\ (\sum_ia_i\pa^i)_{[k]}=a_{k}, \ A^*=\sum_i(-\pa)^ia_i,
\end{eqnarray}
and $A(f)$ denotes the action of $A$ on $f$.

The Kadomtsev-Petviashvili (KP) hierarchy \cite{DJKM} is an important research object in the area of mathematical physics,
which is defined by the following Lax equation
\begin{equation}\label{kphierarchy}
    \frac{\pa L}{\pa t_n}=[(L^{n})_+,L],\ \ n=1,2,3,\cdots,
\end{equation}
with the Lax operator $L$ given by
\begin{equation}\label{laxoperator}
    L=\pa+u_1\pa^{-1}+u_2\pa^{-2}+\cdots,
\end{equation}
where the coefficient functions $u_i$
are all the functions of the time variables
$t=(t_1=x,t_2,t_3,\cdots)$.
The Lax operator (\ref{laxoperator}) can be generated by the dressing operator $\Phi=1+\sum_{k=1}^\infty a_k\pa^{-k}$ in the following way:
\begin{equation}\label{dressingoperator}
    L=\Phi\pa \Phi^{-1}
\end{equation}
Then the Lax equation (\ref{ckphierarchy}) can also be expressed as
Sato's equation
\begin{equation}\label{satoquation}
    \frac{\pa \Phi}{\pa t_n}=-(L^n)_-\Phi,\quad n=1,2,3,\cdots.
\end{equation}
Another important object is the Baker-Akhiezer (BA) wave
function $\psi_{BA}(t,\lambda)$ defined via:
\begin{equation}\label{wavefunction}
\psi_{BA}(t,\lambda)=\Phi(e^{\xi(t,\lambda)})=\phi(t,\lambda)e^{\xi(t,\lambda)}
\end{equation}
with $\xi(t,\lambda)\equiv\sum_{k=1}^\infty t_{k} \lambda^{k}$ and $\phi(t,\lambda)=1+\sum_{i=1}^\infty
a_i(t)\lambda^{-i}$, which satisfies
\begin{eqnarray}
L(\psi_{BA}(t,\lambda))=\lambda\psi_{BA}(t,\lambda),\quad\pa_{t_n} \psi_{BA}(t,\lambda)=(L^n)_+(\psi_{BA}(t,\lambda)),\quad n=1,2,3,\cdots.\label{waveinv}
\end{eqnarray}
The adjoint BA function $\psi_{BA}^*(t,\lambda)$ is introduced through the following way
\begin{equation}\label{adwave}
\psi_{BA}^*(t,\lambda)=\Phi^{*-1}(e^{-\xi(t,\lambda)}).
\end{equation}
The KP hierarchy can also be expressed in terms of a single function called
the tau function $\tau(t)$\cite{DJKM}, which is related with the wave function
in the way below,
\begin{eqnarray}
\psi_{BA}(t,\lambda)=\frac{\tau(t_1-\frac{1}{\lambda},t_2-\frac{1}{2\lambda^2},t_3-\frac{1}{3\lambda^3},\cdots)}{\tau(t_1,t_2,t_3,\cdots)} e^{\xi(t,\lambda)}.\label{kptau}
\end{eqnarray}
Because of the existence of the tau function, many important results in the KP hierarchy can be considered in terms of the tau function, such as
the flow equation, Hirota's bilinear equation and algebraic constraint.
KP hierarchy has two famous
sub-hierarchies\cite{DJKM,djkm1981}: the BKP hierarchy and the CKP hierarchy. Just like the KP hierarchy, the BKP hierarchy also owns
one single tau function, which bring much convenience to the study of the BKP hierarchy.

The CKP hierarchy \cite{djkm1981} is a reduction of the KP hierarchy through the constraint on
$L$ given by (\ref{laxoperator}) as
\begin{equation}\label{CKPconst}
L^*=-L, \end{equation}
then $L$ is called the Lax operator of the CKP hierarchy, and the associated Lax equation of the
CKP hierarchy is
\begin{equation}\label{ckphierarchy}
    \frac{\pa L}{\pa t_n}=[(L^{n})_+,L],\ \ n=1,3,5,\cdots,
\end{equation}
which compresses all even flows, i.e., the Lax equation of the CKP hierarchy has only odd flows. The CKP
constraint (\ref{CKPconst}) on the correspond dressing operator $\Phi$ will be $\Phi^*=\Phi^{-1}$. And thus
in the CKP hierarchy $\psi_{BA}^*(t,\lambda)=\Phi^{*-1}(e^{-\xi(t,\lambda)})=\Phi(e^{-\xi(t,\lambda)})=\psi_{BA}(t,-\lambda)$.
So in the CKP hierarchy, it is enough to only study the wave function $\psi_{BA}(t,\lambda)$.
The CKP hierarchy (\ref{ckphierarchy}) is equivalent to the following bilinear equation:
\begin{equation}\label{cbilin}
  \int d\lambda \psi_{BA}(t,\lambda)\psi_{BA}(t',-\lambda)=0,
\end{equation}
where $\int
d\lambda\equiv\oint_\infty\frac{d\lambda}{2\pi i}=
{\rm Res}_{\lambda=\infty}$. By now, the CKP hierarchy has
attracted many researches\cite{van1,van2,van3,wucz,loris1999,loris2001,hejhep,1he2007,2he2007,hejnmp,tian2011}.

In contrast to the KP and the BKP cases, there seems not a single tau function to describe the CKP hierarchy
 in the form of Hirota bilinear equations (that is, Hirota's equations are no longer of the type $P(D)\tau\cdot\tau=0$)\cite{djkm1981}.
The existence of this kind of tau function for the CKP hierarchy is
a long-standing problem. Much work has been done in this field\cite{van1,van2,van3,wucz}.
Since the existence of tau function of the CKP hierarchy is not proved,
many important results on the Lax operator and dressing operator of the
 CKP hierarchy can not be transferred to the tau
function, until Chang and Wu \cite{wucz} introduces a kind of tau function $\tau_c(t)$
for the CKP case, which is related with the wave function in the following way
\begin{eqnarray}
\psi_{BA}(t,\lambda)&=&\sqrt{\varphi(t,\lambda)}\frac{\tau_c(t-2[\lambda^{-1}])}{\tau_c(t)}e^{\xi(t,\lambda)}\nonumber\\
&=&\left(1+\frac{1}{\lambda}\pa_x{\rm log}\frac{\tau_c(t-2[\lambda^{-1}])}{\tau_c(t)}\right)^{1/2}\frac{\tau_c(t-2[\lambda^{-1}])}{\tau_c(t)}e^{\xi(t,\lambda)},\label{ctauexpression}
\end{eqnarray}
where $[\lambda^{-1}]=(\lambda^{-1},\frac{1}{3}\lambda^{-3},\cdots)$ and $\varphi(t,\lambda)=\phi(t,\lambda)\phi(t-2[\lambda^{-1}],-\lambda)$.
Note that the relation between the tau function and the wave function
is different from the cases of KP and BKP, because there is a square-root factor depending on the tau function.
And the relation of the new CKP tau function $\tau_c$ to the (C-reduced) KP tau function $\tau$ is showed in the Appendix \ref{apptau}.

The ``ghost" symmetry\cite{oevel1993,oevel1994,1OC98,2OC98,aratyn1998}, sometimes called the squared eigenfunction symmetry\cite{oevel1993},
is one of the most important symmetries in the integrable system, which is defined through the squared eigenfunctions\cite{oevel1993}. The usage of
the squared eigenfunction to construct the symmetry flows can be traced back to \cite{orlov1988,orlov1991,orlov1993}, where operators $\psi\partial^{-1} \psi^*$ were introduced to
construct $L-A$ pairs for symmetry flows ($\psi$ and $\psi^*$ being wave functions respectively of $L$ and $L^*$ operators).
And similar symmetry flows are also studied from the Hamiltonian point of view \cite{orlov1996}.
The ``ghost" symmetry can be used to define the new integrable system, such as the
symmetry constraint\cite{oevel1994,2OC98,Cheng91,SS91,KS92,Cheng92,LorisWillox99,Tu2011} and the extended integrable systems\cite{1Zeng08, 2Zeng08},
and investigate the additional symmetry\cite{OS86,ASM95,D95,takasaki1996,Tu07,1he2007,tian2011}. The ``ghost" symmetry
has attracted many researches recently. For example, recently the ``ghost" symmetries for the BKP hierarchy\cite{Cheng10},
 discrete KP hierarchy\cite{Li12}, the Toda lattice hierarchy \cite{cheng2012} and its B type and C type cases \cite{cheng2013} are all studied.

In this paper, firstly starting from the bilinear identity for the CKP hierarchy, the spectral representation for the eigenfunction is established.
Upon the basis of the spectral representation, the expression of the squared eigenfunction potential (SEP) for the eigenfunction and the wave function
is derived, and further all other SEPs are also obtained. Then the ``ghost" symmetry of the CKP hierarchy is constructed
by its action on the Lax operator and the dressing operator. At last, since the existence of the tau function $\tau_c(t)$,
the action of the ``ghost" symmetry is transfered to the tau function $\tau_c(t)$, which has never been studied before.

This paper is organized in the following way. In Section 2, the SEP for the CKP hierarchy is
investigated. In Section 3, we study the ``ghost" symmetry in the CKP hierarchy. At last, some conclusions
and discussions are given in section 4.

\section{SEP for the CKP Hierarchy}
If the functions $q(t)$ and $r(t)$ satisfy
\begin{equation}\label{eigenfunction}
  \frac{\pa q}{\pa{t_n}}=L^n_+(q),\quad\frac{\pa r}{\pa{t_n}}=-(L^n)_+^*(r),\quad n=1,3,5,\cdots,
\end{equation}
then they are called the eigenfunction and the adjoint eigenfunction of the CKP hierarchy respectively.
Obviously according to (\ref{waveinv}), $\psi_{BA}(t,\lambda)$ is the eigenfunction. By the CKP constraint (\ref{CKPconst}),
$L^n_+=-(L^n)_+^*$ for $n\in \mathbb{Z}_+^{\rm odd}$. Thus any adjoint eigenfunction $r(t)$ can be viewed as the eigenfunction and vice versus.

For an arbitrary pair of
eigenfunctions $q_1(t)$ and $q_2(t)$ of the CKP hierarchy, there exists the function $S(q_1(t),q_2(t))$,
called the squared eigenfunction potential (SEP)\cite{oevel1993},
which is determined by the following characteristics:
\begin{equation}\label{sep}
    \frac{\pa}{\pa t_n}S\left(q_1(t),q_2(t)\right)={\rm Res}\left(\pa^{-1}q_2(t)L^n_+q_1(t)\pa^{-1}\right),
n=1,3,5,\cdots,
\end{equation}
These flows are compatible\cite{oevel1993}, namely,
\begin{equation}\label{sepcompatible}
  \frac{\pa}{\pa{t_n}}\frac{\pa}{\pa{t_m}}S(q_1,q_2)=\frac{\pa}{\pa{t_m}}\frac{\pa}{\pa{t_n}}S(q_1,q_2),\quad m,n=1,3,5,\cdots,
\end{equation}
Thus this definition is reasonable. The predecessor of SEP is in fact the
Cauchy-Baker-Akhiezer kernel introduced in \cite{orlov1989}. Note that $S(q_1(t),q_2(t))$ can be up to a constant.
In particular, for $n=1$
\begin{equation}\label{sep2}
    \pa_xS(q_1(t),q_2(t))=q_1(t)q_2(t).
\end{equation}

There are two important properties showed below.
\begin{lemma}
If $q_1$ and $q_2$ are two eigenfunctions of the {\rm CKP}
hierarchy, then
\begin{equation}\label{propertysepckp1}
    S(q_1,q_2)=S(q_2,q_1)
\end{equation}
\end{lemma}
\begin{proof} Assume $n\in \mathbb{Z}_+^{\rm odd}$, then
\begin{eqnarray*}
\frac{\pa}{\pa t_n}S(q_1,q_2)&=&{\rm Res}(\pa^{-1}q_2 L^n_+ q_1 \pa^{-1})\\
&=&-{\rm Res}(\pa^{-1}q_1 (L^n)_+^* q_2 \pa^{-1})\\
&=&{\rm Res}(\pa^{-1}q_1 L^n_+ q_2 \pa^{-1})\\
&=&\frac{\pa}{\pa t_n}S(q_2,q_1),
\end{eqnarray*}
where in the second identity, the relation ${\rm Res} A=-{\rm Res} A^*$ is used. And in the third identity,
we have used the relation $L^n_+=-(L^n)_+^*$ derived by the CKP constraint (\ref{CKPconst}).
\end{proof}
\begin{lemma} Assume $q(t)$ is the eigenfunction of the {\rm CKP} hierarchy, then
\begin{equation}\label{propertysepckp2}
    S(q(t),\psi_{BA}(t,\lambda))=e^{\xi(t,\lambda)}\left(q(t)\lambda^{-1}+\mathcal{O}(\lambda^{-2})\right)
\end{equation}
\end{lemma}
\begin{proof}It can be proved in the same way as Lemma 3 in \cite{Cheng10}.
\end{proof}

Then the relation of the eigenfunction and the wave function can be found in the following proposition.
\begin{proposition}
For any eigenfunction $q(t)$ of the {\rm CKP} hierarchy,
\begin{equation}\label{spectralrepr1}
    q(t)=-\int d\lambda\psi_{BA}(t,\lambda)S(q(t'),\psi_{BA}(t',-\lambda)).
\end{equation}
 In other words, $q(t)$ owns a spectral representation in the following form
\begin{equation}\label{spectralrepr2}
    q(t)=\int d\lambda\rho(\lambda)\psi_{BA}(t,\lambda),
\end{equation}
with spectral densities given by
$\rho(\lambda)=-S(q(t'),\psi_{BA}(t',-\lambda))$.
\end{proposition}
\begin{proof}
Denote the right hand side of (\ref{spectralrepr1}) as $I(t,t')$. Then according to the bilinear identity (\ref{cbilin})
of the CKP hierarchy, one can find that $\pa_{t'_m}I(t,t')=0$ for $m\in \mathbb{Z}_+^{\rm odd}$. Thus $I(t,t')=f(t)$. By
considering (\ref{propertysepckp2}),
\begin{equation*}\label{provespectral}
  I(t,t'=t)=\int d\lambda \psi_{BA}(t,\lambda)e^{-\xi(t,\lambda)}(q(t)\lambda^{-1}+\mathcal{O}(\lambda^{-2}))=q(t)
\end{equation*}
\end{proof}
\noindent {\bf Remark 1:} From (\ref{adwave}) and (\ref{spectralrepr1}), we can know that
\begin{equation}\label{qspectralrepr}
  q(t)=-\int d\lambda\psi_{BA}(t,\lambda)S(q(t'),\psi_{BA}^*(t',\lambda)),
\end{equation}
which is consistent with the ordinary KP hierarchy \cite{aratyn1998}.\\
{\bf Remark 2:} Since any adjoint eigenfunction can be viewed as the eigenfunction in the CKP case,
the spectral representation of the adjoint eigenfunction $r(t)$ can be obtained through (\ref{spectralrepr1}),
\begin{eqnarray}
 r(t)&=&-\int d\lambda\psi_{BA}(t,\lambda)S(r(t'),\psi_{BA}(t',-\lambda))\nonumber\\
 &=& \int d\lambda\psi_{BA}(t,-\lambda)S(r(t'),\psi_{BA}(t',\lambda))\ \ \text{letting $\lambda\rightarrow -\lambda$}\nonumber\\
 &=& \int d\lambda\psi_{BA}^*(t,\lambda)S(\psi_{BA}(t',\lambda),r(t'))\ \ \text{using (\ref{adwave}) and (\ref{propertysepckp1})}, \label{rspectralrepr}
\end{eqnarray}
which is also consistent with the case of the ordinary KP hierarchy \cite{aratyn1998}.

Because of (\ref{propertysepckp2}), we can represent $S(q(t),\psi_{BA}(t,-\lambda))$ as $K(t,\lambda)e^{-\xi(t,\lambda)}$
(with $K(t,\lambda)=-q(t)\lambda^{-1}+\mathcal{O}(\lambda^{-2}$)). Then by exchanging $t$ and $t'$ and letting
$t'=t+2[k^{-1}]$, the relation (\ref{spectralrepr1}) becomes,
\begin{eqnarray}
&&q(t+2[k^{-1}])\nonumber\\
&=&-\int d\lambda \sqrt{\varphi(t+2[k^{-1}],\lambda)}\frac{\tau_c(t+2[k^{-1}]-2[\lambda^{-1}])}{\tau_c(t+2[k^{-1}])}K(t,\lambda)\left(-1+\frac{2}{1-\lambda/k}\right)\nonumber\\
&=&-q(t)-2k\sqrt{\varphi(t+2[k^{-1}],k)}\frac{\tau_c(t)}{\tau_c(t+2[k^{-1}])}K(t,k)\nonumber
\end{eqnarray}
which leads to
\begin{equation}\label{Ktkexpr}
  K(t,k)=\frac{-1}{2k\sqrt{\varphi(t,-k)}}\left(q(t+2[k^{-1}])+q(t)\right)\frac{\tau_c(t+2[k^{-1}])}{\tau_c(t)},
\end{equation}
where we have used the relations $\varphi(t,\lambda)=1+\mathcal{O}(\lambda^{-2})$ and $\varphi(t+2[k^{-1}],k)=\varphi(t,-k)$ derived by the definition of $\varphi(t,k)$.
Thus
\begin{equation}\label{sepqpsi}
  S(q(t),\psi_{BA}(t,-\lambda))=\frac{-1}{2\lambda\varphi(t,-\lambda)}\left(q(t+2[\lambda^{-1}])+q(t)\right)\psi_{BA}(t,-\lambda)
\end{equation}
Further, all the expressions of the squared eigenfunction potential can be obtained in the following proposition.
\begin{proposition}\label{expressionsep}If $q(t)$, $q_1(t)$ and $q_2(t)$ are two eigenfunctions of the {\rm CKP} hierarchy, then
\begin{eqnarray}
  S(\psi_{BA}(t,\mu),\psi_{BA}(t,\lambda))&=&\frac{1}{2\lambda\varphi(t,\lambda)}\left(\psi_{BA}(t-2[\lambda^{-1}],\mu)+\psi_{BA}(t,\mu)\right)\psi_{BA}(t,\lambda),\label{wavewavesep} \\
  S(q(t),\psi_{BA}(t,\lambda))&=&\frac{1}{2\lambda\varphi(t,\lambda)}\left(q(t+2[\lambda^{-1}])+q(t)\right)\psi_{BA}(t,\lambda)
  \label{qwavesep},\\
  S(q_1(t),q_2(t)) &=& \int\int d\lambda d\mu
  \rho_1(\mu) \rho_2(\lambda)S\left(\psi_{BA}(t,\mu),\psi_{BA}(t,\lambda)\right)\label{phipsisep}.
\end{eqnarray}
\end{proposition}

\section{The ``ghost" symmetry of the CKP hierarchy}
Given a set of eigenfunctions $\left\{q_{1i},q_{2i}\right\}_{i\in\alpha}$, the ``ghost" flow of the CKP
hierarchy is defined through its actions on $L$ and $\Phi$ as follows,
\begin{equation}\label{ghsymckp}
\pa_\alpha L\equiv[\sum_{i\in\alpha}(q_{1i}\pa^{-1}q_{2i}+q_{2i}\pa^{-1}q_{1i}),L],\ \ \ \pa_\alpha \Phi
\equiv\sum_{i\in\alpha}(q_{1i}\pa^{-1}q_{2i}+q_{2i}\pa^{-1}q_{1i})\Phi
\end{equation}
and the corresponding action on the eigenfunction $q(t)$ is
\begin{equation}
 \pa_\alpha q= \sum_{i\in\alpha}\left(q_{1i}S(q_{2i},q)+q_{2i}S(q_{1i},q)\right).
\end{equation}
The examples of the ``ghost" flows can be found in Appendix \ref{example}.

Next, we need to check the consistence of $\pa_\alpha$ with the CKP constraint (\ref{CKPconst}) and $[\pa_\alpha, \pa_{t_n}]=0$, that is to say,
show $\pa_\alpha$ is indeed the symmetry of the CKP hierarchy.

\begin{proposition}
$\pa_\alpha$ is consistent with the {\rm CKP} constraint (\ref{CKPconst}), that is, $(\pa_\alpha L)^*+\pa_\alpha L=0$.
\end{proposition}
\begin{proof}
Firstly, denote $A=\sum_{i\in\alpha}(q_{1i}\pa^{-1}q_{2i}+q_{2i}\pa^{-1}q_{1i})$ for convenience. Then by (\ref{CKPconst})
\begin{equation*}\label{ckpcondition}
    (\pa_\alpha L)^*+\pa_\alpha L=[A,L]^*+[A,L]=-[A^*,L^*]+[A,L]=[A^*+A,L]=0,
\end{equation*}
where $A^*+A=0$ is obvious.
\end{proof}

\begin{proposition}
\begin{equation}\label{commwithflow}
 [\pa_\alpha,\pa_{t_n} ]=0,\ \ n\in \mathbb{Z}_+^{\rm odd}
\end{equation}
\end{proposition}
\begin{proof}Firstly, by (\ref{ckphierarchy}) and (\ref{ghsymckp})
\begin{eqnarray*}
[\pa_\alpha,\pa_{t_n}]L&=&\pa_\alpha [L^n_+,L]-\pa_{t_n}[A,L]\\
&=& [[A,L^n]_+,L]+[L^n_+,[A,L]]-[\pa_{t_n}A,L]-[A,[L^n_+,L]]\\
&=& [[A,L^n]_+,L]-[[A,L^n_+],L]-[\pa_{t_n}A,L]\\
&=& [[A,L^n_+]_-+\pa_{t_n}A,L],
\end{eqnarray*}
where we have used the Jacobi relation in the third identity and $[A,L^n]_+=[A,L^n_+]_+$ in the fourth identity. Thus it is only to show
\begin{equation}\label{patnA}
  \pa_{t_n}A=-[A,L^n_+]_-.
\end{equation}
In fact, according to (\ref{eigenfunction})
\begin{eqnarray*}
&&[A,L^n_+]_-\\
&=&\left(\sum_{i\in\alpha}(q_{1i}\pa^{-1}q_{2i}+q_{2i}\pa^{-1}q_{1i}) L^n_+\right)_--\left(L^n_+\sum_{i\in\alpha}(q_{1i}\pa^{-1}q_{2i}+q_{2i}\pa^{-1}q_{1i})\right)_-\\
&=&\sum_{i\in\alpha}\Big(q_{1i}\pa^{-1}(L^*)^n_+(q_{2i})+q_{2i}\pa^{-1}(L^*)^n_+(q_{1i})\Big)- \sum_{i\in\alpha}\Big(L^n_+(q_{1i})\pa^{-1}q_{2i}+L^n_+(q_{2i})\pa^{-1}q_{1i}\Big)\\
&=&-\sum_{i\in\alpha}\Big(q_{1i}\pa^{-1}L^n_+(q_{2i})+q_{2i}\pa^{-1}L^n_+(q_{1i})+L^n_+(q_{1i})\pa^{-1}q_{2i}+L^n_+(q_{2i})\pa^{-1}q_{1i}\Big)\\
&=&-\pa_{t_n}A,
\end{eqnarray*}
where the following relations have been applied
\begin{equation}\label{relationproof}
(F_+\pa^{-1})_-=F_{[0]}\pa^{-1},\ \ (\pa^{-1}F_+)_-=(F^*)_{[0]}\pa^{-1},
\end{equation}
with $F$ a pseudo differential operator.
\end{proof}

From the propositions above, we can see that the squared eigenfunction flow $\pa_\alpha$ is indeed a kind of symmetry for the CKP hierarchy. Next, let's investigate the action
of $\pa_\alpha$ on the tau function. Before this, the lemma \cite{wucz} below is needed.
\begin{lemma}
\begin{equation}\label{resdressing}
{\rm Res}\ \Phi=a_1(t)=-2\pa_x{\rm log}\tau_c(t).
\end{equation}
\end{lemma}
\begin{proof}
From (\ref{wavefunction}) and (\ref{ctauexpression}), we can obtain
\begin{eqnarray}
\phi(t,\lambda)=\sqrt{\varphi(t,\lambda)}\frac{\tau_c(t-2[\lambda^{-1}])}{\tau_c(t)}.\label{lemmaproof}
\end{eqnarray}
By noting $\sqrt{\varphi(t,\lambda)}=1+\mathcal{O}(\lambda^{-2})$ and $\phi(t,\lambda)=1+\sum_{i=1}^\infty
a_i(t)\lambda^{-i}$, the comparison of coefficients of $\lambda^{-1}$ for the both side of (\ref{lemmaproof})
will lead to (\ref{resdressing}).
\end{proof}
\begin{proposition}
\begin{equation}\label{paactontau}
  \pa_\alpha \tau(t)=-\sum_{i\in\alpha}S(q_{1i}(t),q_{2i}(t))\tau(t).
\end{equation}
\end{proposition}
\begin{proof}
  By taking the residue for the both sides of $\pa_\alpha \Phi
=\sum_{i\in\alpha}(q_{1i}\pa^{-1}q_{2i}+q_{2i}\pa^{-1}q_{1i})\Phi$, one can obtain
\begin{equation}\label{paactona1}
  \pa_\alpha a_1(t)=2\sum_{i\in\alpha}q_{1i}q_{2i}.
\end{equation}
Then the applications of (\ref{resdressing}) and (\ref{sep2}) will lead to (\ref{paactontau}).
\end{proof}
We conclude this section with two Remarks.\\
{\bf Remark 3:} Since the generating function of the additional symmetries for the CKP hierarchy is represented as\cite{1he2007}
\begin{equation}\label{addsymgen}
  Y(\lambda,\mu)=\psi_{BA}(t,-\lambda)\pa^{-1}\psi_{BA}(t,\mu)+\psi_{BA}(t,\mu)\pa^{-1}\psi_{BA}(t,-\lambda),
\end{equation}
thus the ``ghost" symmetry generated by $\psi_{BA}(t,\mu)$ and $\psi_{BA}(t,-\lambda)$ can be viewed as
the generating function of the additional symmetries.\\
{\bf Remark 4:} The constrained CKP hierarchy \cite{loris1999} is
just to identify $\pa_\alpha$ with $-\pa_{t_{2k+1}}$, i.e.
\begin{equation}\label{cCKPlax1}
    (L^{2k+1})_-=\sum_{i\in \alpha}(q_{1i}\pa^{-1}q_{2i}+q_{1i}\pa^{-1}q_{2i}),
\end{equation}
or
\begin{equation}\label{cCKPtau}
    \pa_{t_{2k+1}}\tau(t)=
\sum_{i\in\alpha}S(q_{1i}(t),q_{2i}(t))\tau(t).
\end{equation}
This observation provides a simple mathematical explanation of the symmetry constraint of the CKP hierarchy.
In Appendix \ref{example}, we investigate the example of $\pa_\alpha=-\pa_x$. Besides the above reduction of the
``ghost" symmetry to $1+1$ dimensional equations, there is another important reduction from rational symmetry\cite{kri}.
And the relation of these two approaches was considered in \cite{orlov1996}.

\section{Conclusion and Discussion}
In this paper, in order to get the expression of the SEPs, we start from the bilinear identity of the CKP
hierarchy, and establish the spectral representation of the eigenfunction in Proposition 3. The expression of the SEP
for the eigenfunction and the wave function is derived from the spectral representation. And further
all the other expressions in Proposition 4.  Then the ``ghost" symmetry in the CKP hierarchy is constructed
by its action on the Lax operator and the dressing operator (see (\ref{ghsymckp})). At last, the corresponding action on the
tau function is obtained in Proposition 7. The possible applications of the ``ghost" symmetry in the additional
symmetry and the symmetry constraint are also discussed in Remark 3 and 4.

Though the existence of a sole tau function for the CKP hierarchy is obtained, there are still many results on the Lax operator
and the dressing operator which can not be shifted to the new tau function becuase of the complicated square root in eq.(\ref{ctauexpression}).
For example, the algebraic constraints \cite{wucz} on the new tau function are not thoroughly solved and worthy of further study.
In Appendix \ref{apptau}, the relation of the new CKP tau function $\tau_c$ to the (C-reduced) KP tau function
$\tau$ is obtained (see (\ref{taurelation})).

In Appendix \ref{example}, the mKdV equation arises in the constrained CKP hierarchy \cite{loris1999,orlov1991}(see (\ref{mkdv}) and (\ref{5mkdv}))
when considering the reduction of $\pa_\alpha=-\pa_x$. Since the mKdV equation is the well-studied object, it will be very
interesting to compare with the results about mKdV equation in the research of the ``ghost" symmetry for the CKP hierarchy.
And further reduction may lead to the Newton's equation for a pair of particles\cite{OR1993}. These problems are very interesting in the study
of the CKP hierarchy, and we will consider these problems in the later paper.
\appendix
\section{ The Relation between $\tau_c$ and $\tau$ }\label{apptau}
There are two kinds of tau functions for the CKP hierarchy: one is $\tau$ inherited from the KP hierarchy, the other is
$\tau_c$ introduced by Chang and Wu in \cite{wucz}. These two tau functions relate the wave function $\psi_{BA}(t,\lambda)$ of the CKP hierarchy in the following way,
\begin{eqnarray}
\psi_{BA}(t,\lambda)&=&\left(1+\frac{1}{z}\pa_x{\rm log}\frac{\tau_c(t-2[\lambda^{-1}])}{\tau_c(t)}\right)^{1/2}\frac{\tau_c(t-2[\lambda^{-1}])}{\tau_c(t)}e^{\xi(t,\lambda)}\nonumber\\
&=&\frac{\tau(t_1-\frac{1}{\lambda},-\frac{1}{2\lambda^2},t_3-\frac{1}{3\lambda^3},\cdots)}{\tau(t_1,0,t_3,\cdots)}e^{\xi(t,\lambda)}.\label{tau}
\end{eqnarray}

Denote $g(t)={\rm log}\tau_c(t)={\rm log}\tau_c(t_1,t_3,t_5,\cdots)$ and $f(t;\lambda)={\rm log}\tau(t_1,-\frac{\lambda^2}{2},t_3,-\frac{\lambda^4}{4}, t_5,\cdots)$, then by
(\ref{tau}),
\begin{eqnarray}
&&\frac{1}{2}{\rm log}\left(1+\lambda\pa_x g(t_1-2\lambda,t_3-\frac{2}{3}\lambda^3,\cdots)-\lambda\pa_x g(t)\right)+ g(t_1-2\lambda,t_3-\frac{2}{3}\lambda^3,\cdots)-g(t)\nonumber\\
&=& f(t_1-\lambda, t_3-\frac{1}{3}\lambda^3,\cdots;\lambda)-f(t;0).\label{relat2}
\end{eqnarray}
Next if let $t_j\rightarrow t_j+\frac{\lambda^j}{j}$, eq.(\ref{relat2}) will become into
\begin{eqnarray}
&&\frac{1}{2}{\rm log}\left(1+\lambda\pa_x g(t_1-\lambda,t_3-\frac{\lambda^3}{3},\cdots)-\lambda\pa_xg(t_1+\lambda,t_3+\frac{\lambda^3}{3},\cdots)\right)\nonumber\\
&&+ g(t_1-\lambda,t_3-\frac{\lambda^3}{3},\cdots)-g(t_1+\lambda,t_3+\frac{\lambda^3}{3},\cdots)\nonumber\\
&=& f(t;\lambda)-f(t_1+\lambda, t_3+\frac{1}{3}\lambda^3,\cdots;0).\label{relat3}
\end{eqnarray}
Further if let $\lambda\rightarrow-\lambda$, then eq.(\ref{relat3}) is transferred to
\begin{eqnarray}
&&\frac{1}{2}{\rm log}\left(1-\lambda\pa_x g(t_1+\lambda,t_3+\frac{\lambda^3}{3},\cdots)+\lambda\pa_xg(t_1-\lambda,t_3-\frac{\lambda^3}{3},\cdots)\right)\nonumber\\
&&+ g(t_1+\lambda,t_3+\frac{\lambda^3}{3},\cdots)-g(t_1-\lambda,t_3-\frac{\lambda^3}{3},\cdots)\nonumber\\
&=& f(t;\lambda)-f(t_1-\lambda, t_3-\frac{1}{3}\lambda^3,\cdots;0),\label{relat4}
\end{eqnarray}
where note that $f(t;\lambda)=f(t;-\lambda)$.\\
At last the subtraction of  (\ref{relat4}) from (\ref{relat3}) will lead to
\begin{eqnarray}
&&2\left(g(t_1-\lambda,t_3-\frac{\lambda^3}{3},\cdots)-g(t_1+\lambda,t_3+\frac{\lambda^3}{3},\cdots)\right)\nonumber\\
&=&f(t_1-\lambda, t_3-\frac{1}{3}\lambda^3,\cdots;0)-f(t_1+\lambda, t_3+\frac{1}{3}\lambda^3,\cdots;0).\label{relat5}
\end{eqnarray}
From (\ref{relat5}), we can know that
\begin{equation}\label{relat6}
  2g(t)=f(t;0)+{\rm const}
\end{equation}
Therefore
\begin{equation}\label{taurelation}
  \tau^2_c(t_1,t_3,t_5,\cdots)={\rm const}\cdot\tau(t_1,0,t_3,0,t_5,\cdots).
\end{equation}
\section{Examples of the ``ghost" flows }\label{example}
Here we list some examples of the ``ghost" flows $\pa_\alpha$ generated by the eigenfunctions $q_1$ and $q_2$ for the CKP hierarchy, that is,
\begin{equation}\label{ghosteqapp}
  \pa_\alpha L=[q_1\pa^{-1}q_2+q_2\pa^{-1}q_1,L],
\end{equation}
where $L$ is the Lax operator of the CKP hierarchy given by (\ref{laxoperator}), and satisfies the CKP constraint (\ref{CKPconst}).

The actions of $\pa_\alpha$ on the the first few terms of the Lax operator $L$ are showed below.
\begin{eqnarray}
\pa_\alpha u_1&=&-2(q_1q_2)_x,\label{ghostu1}\\
\pa_\alpha u_2&=&(q_1q_2)_{xx},\label{ghostu2}\\
\pa_\alpha u_3&=&2(q_1q_2)_xu_1-2q_1q_2u_{1x}-(q_{1}q_{2xx}+q_{2}q_{1xx})_x,\label{ghostu3}\\
\pa_\alpha u_4&=&4(q_1q_2)_xu_2-2q_1q_2u_{2x}-3(q_2q_{1xx}+q_1q_{2xx})u_1\nonumber\\
&&-6q_{1x}q_{2x}u_1+2(q_1q_2u_{1x})_x+(q_1q_{2xxx}+q_2q_{1xxx})_x,\label{ghostu4}\\
\pa_\alpha u_5&=&(q_1q_2)_x(6u_3+2u_{2x}-3u_{1xx})-(q_2q_{1xx}+q_1q_{2xx})(8u_2+3u_{1x})\nonumber\\
&&-16q_{1x}q_{2x}u_2+2q_1q_2(u_{2xx}-u_{3x})+4(q_2q_{1xxx}+q_1q_{2xxx})u_1\nonumber\\
&&+10(q_{1x}q_{2x})_xu_1-(q_2q_{1xxxx}+q_1q_{2xxxx})_x.\label{ghostu5}
\end{eqnarray}

Note that the CKP constraint (\ref{CKPconst}) is equivalent to
\begin{eqnarray}
u_2&=& -\frac{1}{2}u_{1x},\label{u2u1}\\
u_4&=&\frac{1}{4}u_{1xxx}-\frac{3}{2}u_{3x},\label{u4u13}\\
u_6&=&-\frac{1}{2}u_{1xxxxx}+\frac{5}{2}u_{3xxx}-\frac{5}{2}u_{5x},\label{u6u135}\\
&\vdots&\nonumber
\end{eqnarray}
Using these examples, we can find that $\pa_\alpha$ is consistent with the CKP constraint (\ref{CKPconst}) for the first few terms. In fact,
by (\ref{ghostu1})-(\ref{ghostu4}) and (\ref{u2u1})-(\ref{u4u13}),
\begin{eqnarray}
\pa_\alpha u_2&=& -\frac{1}{2}(\pa_\alpha u_{1})_x,\label{symu2u1}\\
\pa_\alpha u_4&=&\frac{1}{4}(\pa_\alpha u_{1})_{xxx}-\frac{3}{2}(\pa_\alpha u_{3})_x.\label{symu4u13}
\end{eqnarray}

According to (\ref{eigenfunction})(\ref{u2u1}) and (\ref{u4u13})
\begin{eqnarray}
\pa_{t_3}q_i&=&q_{ixxx}+3u_1q_{ix}+\frac{3}{2}u_{1x}q_i,\label{t3qi}\\
\pa_{t_5}q_i&=&q_{ixxxxx}+5u_1q_{ixxx}+\frac{15}{2}u_{1x}q_{ixx}+(5u_3+10u_1^2+5u_{1xx})q_{ix}\nonumber\\
&&+(\frac{5}{4}u_{1xxx}+\frac{5}{2}u_{3x}+10u_1u_{1x})q_i, \quad i=1,2 \label{t5qi}
\end{eqnarray}
If letting $\pa_\alpha=-\pa_x$, then from (\ref{ghostu1}) (\ref{ghostu3}) (\ref{t3qi}) and (\ref{t5qi}), one can get
\begin{eqnarray}
\pa_{t_3}q_1&=&q_{1xxx}+9q_1q_2q_{1x}+3q_1^2q_{2x},\nonumber\\
\pa_{t_3}q_2&=&q_{2xxx}+9q_1q_2q_{2x}+3q_2^2q_{1x}.\label{t3q12}
\end{eqnarray}
and
\begin{eqnarray}
\pa_{t_5}q_1&=&q_{1xxxxx}+15q_1q_2q_{1xxx}+30q_{1xx}q_{1x}q_2\nonumber\\
&&+25q_{1xx}q_1q_{2x}+25q_{1x}q_1q_{2xx}+20q_{1x}^2q_{2x}\nonumber\\
&&+80q_{1x}q_1^2q_2^2+5q_1^2q_{2xxx}+40q_1^3q_2q_{2x},\nonumber\\
\pa_{t_5}q_2&=&q_{2xxxxx}+15q_1q_2q_{2xxx}+30q_{2xx}q_1q_{2x}\nonumber\\
&&+25q_{2xx}q_{1x}q_2+25q_{2x}q_2q_{1xx}+20q_{1x}q_{2x}^2\nonumber\\
&&+80q_{2x}q_1^2q_2^2+5q_2^2q_{1xxx}+40q_1q_2^3q_{1x}.\label{t5q12}
\end{eqnarray}
Let $q_1=q_2=q$, (\ref{t3q12}) implies the mKdV equation
\begin{equation}\label{mkdv}
  \pa_{t_3}q=q_{xxx}+12q^2q_{x},
\end{equation}
and (\ref{t3q12}) leads to the 5th order mKdV equation
\begin{equation}\label{5mkdv}
 \pa_{t_5}q=q_{xxxxx}+20q^2q_{xxx}+80q_{xx}q_{x}q+20q_{x}^3+120q^4q_{x}.
\end{equation}
Note that the mKdV  and the 5th order mKdV equations
are also reduced from 1-constrained CKP hierarchy through the third flow and the fifth flow respectively\cite{lcz}.

{\bf Acknowledgments}

{\noindent \small This work is supported by the NSFC (Grant Nos. 11301526 and 11371361) and the Fundamental Research Funds for the
Central Universities (Grant No. 2012QNA45). We thank anonymous referee for his/her useful suggestions on appendices and several references.}

\end{document}